\documentclass{article}

\usepackage{times}
\usepackage{graphicx} % more modern
\usepackage{subfigure} 

\usepackage{natbib}

\usepackage{algorithm}
\usepackage{algorithmic}

\usepackage{amsmath, amssymb, amsthm}

\usepackage{hyperref}

\usepackage[accepted]{icml2016}

\icmltitlerunning{On Greedy Column Subset Selection}

\newtheorem{defin}{Definition}%[section]
\newtheorem{thm}{Theorem}%[section]
\newtheorem{lemma}{Lemma}%[section]
\newtheorem{remark}{Remark}%[section]
\newcommand{\Real}{\mathbb{R}}
\newcommand{\E}{\mathbb{E}}
\newcommand{\argmin}{\text{arg min}}
\newcommand{\argmax}{\text{arg max}}

\newcommand{\Proj}{\text{proj}}
\newcommand{\Span}{\text{span}}
\newcommand{\Pow}{\mathcal{P}}
\newcommand{\eps}{\varepsilon}
\newcommand{\del}{\delta}
\newcommand{\re}{\mathbb{R}}
\newcommand{\norm}[1]{\lVert #1 \rVert}
\newcommand{\calN}{\mathcal{N}}
\newcommand{\plusminus}{\raisebox{.2ex}{$\scriptstyle\pm$}}

\newcommand{\enum}[1]{\begin{enumerate} #1 \end{enumerate}}
\newcommand{\greedy}{\textsc{Greedy}}

\newcommand{\ltlgreedy}{\textsc{Lazier-than-lazy Greedy}}
\newcommand{\distgreedy}{\textsc{Distgreedy}}
\newcommand{\iprod}[1]{\langle #1 \rangle}
\newcommand{\opt}{OPT}
\newcommand{\gcss}{\textsf{GCSS}}
\newcommand{\maxcoverage}{\textsf{MAX-COVERAGE}}

\begin{document} 

\twocolumn[
\icmltitle{Greedy Column Subset Selection: \\ 
           New Bounds and Distributed Algorithms}

\icmlauthor{Jason Altschuler}{jasonma@princeton.edu}
\icmladdress{Princeton University,
            Princeton, NJ 08544}
\icmlauthor{Aditya Bhaskara}{bhaskara@cs.utah.edu}
\icmladdress{School of Computing, 50 S. Central Campus Drive,
             Salt Lake City, UT 84112}
\icmlauthor{Gang (Thomas) Fu}{thomasfu@google.com}
\icmlauthor{Vahab Mirrokni}{mirrokni@google.com}
\icmlauthor{Afshin Rostamizadeh}{rostami@google.com}
\icmlauthor{Morteza Zadimoghaddam}{zadim@google.com}
\icmladdress{Google, 76 9th Avenue, New York, NY 10011}

\icmlkeywords{column selection, greedy algorithms, coresets}

\vskip 0.3in
]

\begin{abstract}
The problem of column subset selection has recently attracted a
large body of research, with feature selection serving as one obvious
and important application. Among the techniques that have been
applied to solve this problem, the greedy algorithm has been shown to
be quite effective in practice. However, theoretical guarantees on its performance have not been explored thoroughly, especially
in a distributed setting.  In this paper, we study the greedy
algorithm for the column subset selection problem from a theoretical and
empirical perspective and show its effectiveness in a distributed
setting. In particular, we provide an improved 
approximation guarantee for the greedy algorithm which we show is tight up to a constant factor, and
present the first distributed implementation with provable
approximation factors. We use the idea of randomized composable core-sets, developed recently in the context of submodular maximization. 
Finally, we validate the effectiveness of this distributed algorithm via an
empirical study.
\end{abstract} 

\section{Introduction}
Recent technological advances have made it possible to collect unprecedented amounts of data. However, extracting patterns of information from these high-dimensional massive datasets is often challenging. How do we automatically determine, among millions of measured features (variables), which are informative, and which are irrelevant or redundant? The ability to select such features from high-dimensional data is crucial for computers to recognize patterns in complex data in ways that are fast, accurate, and even human-understandable \cite{Guyon}.

An efficient method for feature selection receiving increasing attention is Column Subset Selection (CSS). CSS is a constrained low-rank-approximation problem that seeks to approximate a matrix (e.g. instances by features matrix) by projecting it onto a space spanned by only a few of its columns (features). Formally, given a matrix $A$ with $n$ columns, and a target rank $k < n$, we wish to find a size-$k$ subset $S$ of $A$'s columns such that each column $A_i$ of $A$ ($i \in \{1, \dots, n\}$) is contained as much as possible in the subspace $\Span(S)$, in terms of the Frobenius norm:
\begin{align*}
\argmax_{S \text{ contains $k$ of $A$'s columns}} \sum_{i=1}^n \|\Proj(A_i \; | \; \Span(S))\|_2^2
\end{align*}

While similar in spirit to general low-rank approximation, some advantages with CSS include flexibility,
interpretability and efficiency during inference. CSS is an
unsupervised method and does not require labeled data, which is
especially useful when labeled data is sparse.  We note, on the other
hand, unlabeled data is often very abundant and therefore scalable
methods, like the one we present, are often needed.  
%Further, CSS is independent of the learning algorithm that will be
%used on the reduced data (i.e. it is classifier-independent).  This
%implies that the reduced data can easily be used with many learning
%algorithms, since we will not need a different embedded feature
%selector for each learning algorithm.
Furthermore, by subselecting features, as opposed to generating new
features via an arbitrary function of the input features, we keep the
semantic interpretation of the features intact. This is especially
important in applications that require interpretable models.  A third
important advantage is the efficiency of applying the solution CSS
feature selection problem during inference. Compared to PCA or other
methods that require a matrix-matrix multiplication to project input
features into a reduced space during inference time, CSS only requires
selecting a subset of feature values from a new instance vector. This
is especially useful for latency sensitive applications and when the
projection matrix itself may be prohibitively large, for example in
restricted memory settings.

While there have been significant advances in CSS~\cite{Boutsidis1,Boutsidis2,Guruswami}, most of the algorithms are either impractical and not applicable in a distributed setting for large datasets, or they do not have good (multiplicative $1 - \eps$) provable error bounds. 
Among efficient algorithms  studied for the CSS problem is the simple {\em greedy algorithm}, which iteratively selects the best column and keeps it. Recent work shows that it does well in practice and even in a distributed setting~\cite{Farahat1, Farahat2} and admits a performance guarantee \cite{Civril1}. However, the known guarantees depend on an arbitrarily large matrix-coherence parameter, which is unsatisfactory. Also, even though the algorithm is relatively fast, additional optimizations are needed to scale it to datasets with millions of features and instances. %In this paper, we study the performance of the greedy algorithm and a distributed implementation of this algorithm, and provide new guarantees for this algorithm.

\subsection{Our contributions}
Let $A \in \Real^{m \times n}$ be the given matrix, and let $k$ be the target number of columns. Let $\opt_k$ denote the {\em optimal} set of columns, i.e., one that {\em covers} the maximum Frobenius mass of $A$. Our contributions are as follows.

{\em Novel analysis of Greedy.} For any $\eps > 0$, we show that the natural greedy algorithm (Section~\ref{section-2}), after $r = \frac{k}{\sigma_{\min}(\opt_k) \eps}$ steps, gives an objective value that is within a $(1-\eps)$  factor of the optimum. We also give a matching lower bound, showing that $\frac{k}{\sigma_{\min}(\opt_k) \eps}$ is tight up to a constant factor. Here $\sigma_{\min}(\opt_k)$ is the smallest squared singular value of the {\em optimal} set of columns (after scaling to unit vectors).

Our result is similar in spirit to those of~\cite{Civril1, Liberty}, but with an important difference. Their bound on $r$ depends on the {\em least} $\sigma_{\min}(S)$ over \textit{all} $S$ of size $k$, while ours depends on $\sigma_{\min}(\opt_k)$. Note that these quantities can differ significantly. For instance, if the data has even a little bit of redundancy (e.g. few columns that are near duplicates), then there exist $S$ for which $\sigma_{\min}$ is tiny, but the optimal set of columns could be reasonably well-conditioned (in fact, we would {\em expect} the optimal set of columns to be fairly well conditioned).

% for any matrix $A \in \Real^{m \times n}$, positive integer $k \leq n$, and error threshold $\eps > 0$, then
%$$\|\Pi_{A[S_r]}A\|_F^2 \geq (1 - \eps) \|\Pi_{A[OPT_k]}A\|_F^2$$
%for any $r \geq \frac{16k}{\eps \sigma_{min}(OPT_k)}$, where $S_r$ is the set of $r$ selected columns by GREEDY, and $OPT_k = \argmax_{S \subseteq [n] \; : \; |S| \leq k} f(S)$ is the optimal set of $k$ columns. This is a significantly tigher bound than is given in the analysis in \cite{Civril1}, especially because it relies only on the smallest singular value of the optimal solution, not the smallest singular value of any full-rank submatrix of $A$.

{\em Distributed Greedy.} We consider a natural distributed implementation of the greedy algorithm (Section~\ref{section-2}). Here, we show that an interesting phenomenon occurs: even though partitioning the input does not work in general (as in coreset based algorithms), {\em randomly} partitioning works well. This is inspired by a similar result on submodular maximization~\cite{Mirrokni}. % with a provable approximation factor, as shown in Theorem ~\ref{thm:core-set}. The main idea is to partition the columns randomly onto We do so by applying the idea of randomized composable core-sets. In this approach, the data is partitioned into $m$ parts
%at random, and each part is solved separately in parallel to other parts, and finally we combine the outputs of each part, and solve the problem on the union. Interestingly, the underlying composable core-set proof also implies applicability of our algorithm in a randomized streaming model~\cite{}.
Further, our result implies a $2$-pass streaming algorithm for the CSS problem in the {\em random arrival} model for the columns.

%In particular, we show that  $\E[Max_{1 \leq i \leq \ell}\{f(S), f(S_i)\}]$, is $\Omega(\sigma_{min}(OPT)f(OPT))$.

We note that if the columns each have sparsity $\phi$,~\cite{Boutsidis2015} gives an algorithm with total communication of $O(\frac{sk\phi}{\eps} + \frac{sk^2}{\eps^4})$. Their algorithm works for ``worst case'' partitioning of the columns into machines and is much more intricate than the greedy algorithm. In constrast, our algorithm is very simple, and for a random partitioning, the communication is just the first term above, along with an extra $\sigma_{\min}(\opt)$ term. Thus depending on $\sigma_{\min}$ and $\eps$, each of the bounds could be better than the other.

{\em Further optimizations.} We also present techniques to speed up the implementation of the greedy algorithm. We show that the recent result of~\cite{Mirzasoleiman} (once again, on submodular optimization) can be extended to the case of CSS, improving the running time significantly.
%\item We also believe the technical ideas in our analysis of Greedy may be of independent interest, in particular the proof ideas in section 3.3.

We then compare our algorithms (in accuracy and running times) to various well-studied CSS algorithms. (Section 6.)

\subsection{Related Work}
The CSS problem is one of the central problems related to matrix approximation. Exact solution is known to be UG-hard~\cite{Civril2}, and several approximation methods have been proposed over the years. Techniques such as importance sampling \cite{Drineas1, Frieze}, adaptive sampling \cite{Deshpande1}, volume sampling \cite{Deshpande2, Deshpande4}, leverage scores \cite{Drineas-Leverage}, and projection-cost preserving sketches \cite{Cohen} have led to a much better understanding of the problem. \cite{Guruswami} gave the optimal dependence between column sampling and low-rank approximation.
Due to the numerous applications, much work has been done on the implementation side, where adaptive sampling and leverage scores have been shown to perform well. A related, extremely simple algorithm is the greedy algorithm, which turns out to perform well and be scalable \cite{Farahat1, Farahat2}. This was first analyzed by~\cite{Civril1}, as we discussed. %showed approximation guarantees for this  but their bounds suffer from dependence on a ``matrix coherence parameter'', which is  This term can be arbitrarily large and is essentially the maximum condition number over \textit{all} sized $k$ column subset of the matrix $A$. In contrast, we obtain a bound that depends instead on the condition number of the optimum set of $k$ columns, which intuitively one would expect to be large. We further show that this dependence tight up to a constant in Section 2. (TODO: Delete some of the above if redundant with intro?)

There is also substantial literature about distributed algorithms for CSS \cite{Pi, Feldman, Cohen, Farahat3, Farahat4, Boutsidis2015}. In particular, \cite{Farahat3, Farahat4} present distributed versions of the greedy algorithm based on MapReduce. Although they do not provide theoretical guarantees, their experimental results are very promising.

The idea of composable coresets has been applied explicitly or implicitly to several problems~\cite{FeldmanSS13,BalcanEL13,VahabPODS2014}. Quite recently, for some problems in which coreset methods do not work in general, surprising results have shown that randomized variants of them give good approximations~\cite{BarbosaENW15,Mirrokni}. We extend this framework to the CSS problem.

\subsection{Background and Notation}
We use the following notation throughout the paper. The set of integers $\{1, \dots, n\}$ is denoted by $[n]$. For a matrix $A \in \Real^{m \times n}$, $A_j$ denotes the $j$th column ($A_j \in \Real^m$). Given $S \subseteq [n]$, $A[S]$ denotes the submatrix of $A$ containing columns indexed by $S$. The projection matrix $\Pi_A$ projects onto the column span of $A$. % A$ the columnwise projection of $A$ onto the column space of $B$. That is, $(\Pi_B A)_j = A_j - \Proj(A_j \; | \; \Span(\{B_i\}_{i=1}^k))$. It is simple to see that $\Pi_B A = \argmin_{C \text{ in column space of }B}\|A - C\|_F^2$. Let $\langle A, \; B\rangle_F = \sum_{i,j} A_{i,j} B_{i,j}$ denote the Frobenius inner product between two matrices of the same dimensions. 
Let $\norm{A}_F$ denote the Frobenius norm, i.e., $\sqrt{\sum_{i,j} A_{i, j}^2}$. We write $\sigma_{\min}(A)$ to denote the minimum \textit{squared} singular value, i.e., $\inf_{x:\norm{x}_2 = 1} \frac{\|Ax\|_2^2}{\|x\|_2^2}$. We abuse notation slightly, and for a set of vectors $V$, we write $\sigma_{\min}(V)$ for the $\sigma_{\min}$ of the matrix with columns $V$.%  the minimum singular value of the matrix created by concatenating the $v_1, \dots, v_k$.

\iffalse
{\bf Submodular optimization.}  Given a finite set $\Omega$ and a set function $f : 2^{\Omega} \to \Real$, define the marginal gain of adding an element $x \in \Omega$ to a set $S \subseteq \Omega$ by $\Delta(x | S) = f(S \cup \{x\}) - f(S)$. $f$ is said to be submodular if $\Delta(x | S) \geq \Delta(x | T)$ for any subsets $S \subseteq T \subseteq \Omega$ and any element $x \in \Omega \setminus T$. This is a formalization of the well-known economic principle of decreasing marginal utility. $f$ is further said to be nonnegative if $f(S) \geq 0$ for any $S \subseteq \Omega$, and monotonically nondecreasing if $f(S) \leq f(T)$ for any $S \subseteq T \subseteq \Omega$. The theory of maximizing submodular functions subject to a cardinality constraint has been well studied, and has been shown to be NP-hard [Nemhauser and
Wolsey 1978; Feige 1998]. However, it is a key result in combinatorial optimization that a simple greedy algorithm to this problem for nonnegative, monotone nondecreasing submodular functions admits a $1 - \frac{1}{e}$ constant factor approximation [Nemhauser '78].
\fi

\begin{defin}\label{defn:css-problem}
Given a matrix $A \in \Real^{m \times n}$ and an integer $k \le n$, the \textbf{Column Subset Selection (CSS) Problem} asks to find
\[ \argmax_{S \subseteq [n], |S| = k} \norm{\Pi_{A[S]}A}_F^2, \]
i.e., the set of columns that {\em best explain} the full matrix $A$.
\end{defin}

We note that it is also common to cast this as a minimization problem, with the objective being $\norm{A - \Pi_{A[S]} A}_F^2$. While the exact optimization problems are equivalent, obtaining multiplicative approximations for the minimization version could be harder when the matrix is low-rank.

For a set of vectors $V$ and a matrix $M$, we denote
\[ f_M(V)  = \norm{\Pi_V M}_F^2. \]
%We will abuse notation slightly and for $S \subseteq [n]$, write use $f_M(S)$ to denote $f_A(A[S])$. 
\iffalse
In this article, instead of minimizing the unexplained (error) part $\|A - \Pi_{A[S]}A\|_F^2$ of $A$, we maximize the explained part $\|\Pi_{A[S]}A\|_F^2$ of $A$. Formally,
\begin{align}
\argmax_{S \subseteq [n], |S| = k} \|\Pi_{A[S]}A\|_F^2
= \argmin_{S \subseteq [n], |S| = k} \|A - \Pi_{A[S]}A\|_F^2
\end{align}
Thus, a subset of columns that maximizes explanation of $A$ will also minimized the unexplained error.
\fi
Also, the case when $M$ is a single vector will be important. For any vector $u$, and a set of vectors $V$, we write
\[  f_u(V) = \norm{\Pi_V u}_2^2. \]

\begin{remark} \label{rem:not-submodular} Note that $f_M (V)$ can be viewed as the extent to which we can {\em cover} matrix $M$ using vectors $V$.  However, unlike combinatorial covering objectives, our definition is not submodular, or even subadditive.
\iffalse
\begin{defin} \label{f definition matrix}
Given $A \in \Real^{m \times n}$, define the function: $f_A : \Pow(\Real^m) \to \Real$ by: $$f_A(S) = \sum_{j=1}^n f_{A_j}(S)$$ over the columns $A_j \in \Real^m$ of $A$. 
\end{defin}
\fi
As an example, consider covering the following $A$ using its own columns. Here, $f_A(\{A_1, A_2\}) = \|A\|_F^2 > f(\{A_1\}) + f(\{A_2\})$.
\[ A = \left( \begin{array}{ccc}
1 & 0 & 1 \\
1 & -1 & 0 \\
0 & 1 & 1 \end{array} \right)\]
%It is easy to check that . Thus $f$ is not subadditive, and hence also not submodular.
\end{remark}

\iffalse
\subsection{Overview of article}
In section 2.1, we present a ``vanilla'' greedy algorithm GREEDY for choosing $k$ columns of a matrix $A$ that can approximate all of $A$'s other columns linearly. This provides intuition for our proposed algorithm ALG, presented in section 2.2. ALG is a much more efficient version of GREEDY because of three optimizations: (1) an efficient calculation of marginal gain; (2) a random projection to compress the ambient dimension of $A$'s columns; and (3) only looking over a small random subset of all $n$ columns in each iteration. These optimizations make analyzing ALG slightly more involved than GREEDY. To this end, we first analyze GREEDY in section 3, and then use those results to analyze ALG in section 4.
\fi

\section{Greedy Algorithm for Column Selection} \label{section-2}

Let us state our algorithm and analysis in a slightly general form. Suppose we have two matrices $A, B$ with the same number of rows and $n_A$, $n_B$ columns respectively. The $\gcss(A, B, k)$ problem is that of finding a subset $S$ of columns of $B$, that maximizes $f_A(S)$ subject to $|S|=k$.
Clearly, if $B = A$, we recover the CSS problem stated earlier. Also, note that scaling the columns of $B$ will not affect the solution, so let us assume that the columns of $B$ are all unit vectors. The greedy procedure iteratively picks columns of $B$ as follows:

\begin{algorithm} \label{alg:greedy}
\caption{$\greedy$($A \! \in \! \Real^{m \times n_A}$, $B \! \in\!  \Real^{m \times n_B}$, $k \leq n_B$)}
\begin{algorithmic}[1]
\STATE $S \leftarrow \emptyset$
\FOR{$i = 1:k$}
\STATE Pick column $B_j$ that maximizes $f_A(S \cup B_j)$
\STATE $S \leftarrow S \cup \{B_j\}$
\ENDFOR
\STATE Return $S$
\end{algorithmic}
\end{algorithm}

Step (3) is the computationally intensive step in $\greedy$ -- we need to find the column that gives the most {\em marginal gain}, i.e., $f_A(S \cup B_j) - f_A(S)$.  %Note that it suffices to efficiently calculate the marginal gain $f_A(S\cup B_j) - f_A(S)$, since clearly a maximizer of that is also a maximizer of $f_A(S \cup B_j)$.
In Section~\ref{section-5}, we describe different techniques to speed up the calculation of marginal gain, while obtaining a $1-\eps$ approximation to the optimum $f(\cdot)$ value. Let us briefly mention them here. %These are discussed in detail in Section~\ref{sec:speedup}, but briefly enumerated here.

{\em Projection to reduce the number of rows.}  We can left-multiply both $A$ and $B$ with an $r \times n$ Gaussian random matrix.  For $r \ge \frac{k\log n}{\eps^2}$, this process is well-known to preserve $f_A(\cdot)$, for any $k$-subset of the columns of $B$ (see~\cite{Sarlos} or Appendix Section~\ref{app:random-projections} for details).%\footnote{Note that random projection is highly parallelizable.}

{\em Projection-cost preserving sketches.}
Using recent results from \cite{Cohen}, we can project each {\em row} of $A$ onto a random $O(\frac{k}{\eps^2})$ dimensional space, and then work with the resulting matrix. Thus we may assume that the number of columns in $A$ is $O(\frac{k}{\eps^2})$. This allows us to efficiently compute $f_A(\cdot)$.

\iffalse
\textbf{Random projections to reduce the number of rows.} We can project each column of $A$ and $B$ onto a random $O(\frac{k\log (\max(n', n_B))}{\eps^2})$ dimensional space, and then work with the resulting matrices. Thus we may assume that the number of rows in both $A$ and $B$ is $\min\{ m, O(\frac{k\log( \max(n', n_B))}{\eps^2}) \}$, which can be a big improvement. This can be obtained by a union bound on Lemma 10 from \cite{Sarlos}. (Full details in appendix.)
\fi

{\em Lazier-than-lazy greedy.}
\cite{Mirzasoleiman} recently proposed the first algorithm that achieves a constant factor approximation for maximizing submodular functions with a {\em linear} number of marginal gain evaluations. We show that a similar analysis holds for $\gcss$, even though the cost function is not submodular.

We also use some simple yet useful ideas from \cite{Farahat2} to compute the marginal gains (see Section~\ref{section-5}).

\subsection{Distributed Implementation}
We also study a distributed version of the greedy algorithm, shown below (Algorithm~\ref{alg:cs-greedy}). $\ell$ is the number of machines.%, using the idea of randomized composable coresets. If we have $\ell$ machines at our disposal, then the coreset procedure does the following:

\begin{algorithm} \label{alg:cs-greedy}
\caption{$\distgreedy$($A$, $B$, $k$, $\ell$)}
\begin{algorithmic}[1]
\STATE {\em Randomly} partition the columns of $B$ into $T_1, \dots, T_{\ell}$
\STATE (Parallel) compute $S_i \leftarrow \greedy(A, T_i, \frac{32k}{\sigma_{\min}(OPT)})$
\STATE (Single machine) aggregate the $S_i$, and compute $S \leftarrow \greedy(A, \cup_{i=1}^{\ell}S_i, \frac{12k}{\sigma_{\min}(OPT)})$
\STATE Return $\argmax_{S' \in \{S, S_1,\dots, S_{\ell}\}} f_A(S')$
\end{algorithmic}
\end{algorithm}

As mentioned in the introduction, the key here is that the partitioning is done {\em randomly}, in contrast to most results on {\em composable summaries}. % -- in fact we give instances in which any reasonable algorithm is bound to fail if the partitioning is done badly (Appendix Section~\ref{app:rand-part}).
We also note that machine $i$ only sees columns $T_i$ of $B$, but requires evaluating $f_A(\cdot)$ on the full matrix $A$ when running \greedy.\footnote{It is easy to construct examples in which splitting both $A$ and $B$ fails badly.} The way to implement this is again by using projection-cost preserving sketches. (In practice, keeping a small sample of the columns of $A$ works as well.) The sketch is first passed to all the machines, and they all use it to evaluate $f_A(\cdot)$.

%TODO (Aditya): Add comments about why we need to evaluate the computation on the whole matrix $A$. Give the hard example. Motivate projection-cost preserving sketches.

We now turn to the analysis of the single-machine and the distributed versions of the greedy algorithm.

\section{Peformance analysis of GREEDY} \label{section-3}

The main result we prove is the following, which shows that by taking only slightly more than $k$ columns, we are within a $1 -\eps$ factor of the optimal solution of size $k$.

\begin{thm} \label{thm:greedy-main}
Let $A \in \Real^{m \times n_A}$ and $B \in \Real^{m \times n_B}$. Let $OPT_k$ be a set of columns from $B$ that maximizes $f_A(S)$ subject to $|S| = k$.  Let $\eps > 0$ be any constant, and let $T_r$ be the set of columns output by $\greedy(A, B, r)$, for $r = \frac{16k}{\eps \sigma_{\min}(OPT_k)}$. Then we have%Fix any matrix $A \in \Real^{m \times n}$, positive integer $k \leq n$, and error threshold $\eps > 0$. Then for any $r \geq \frac{16k}{\eps \sigma_{min}(OPT_k)}$, we have that:
$$f_A(T_r) \geq (1 - \eps) f_A(OPT_k).$$
%where $T_r$ is the set of $r$ selected columns by GREEDY, and $OPT_k = \argmax_{S \subseteq [n] \; : \; |S| \leq k} f(S)$ is the optimal set of $k$ columns.
\end{thm}

We show in Appendix Section \ref{app:tight-ex} that this bound is tight up to a constant factor, with respect to $\eps$ and $\sigma_{\min}(OPT_k)$. Also, we note that $\gcss$ is a harder problem than $\maxcoverage$, implying that if we can choose only $k$ columns, it is impossible to approximate to a ratio better than $(1-\frac{1}{e}) \approx 0.63$, unless P=NP. (In practice, $\greedy$ does much better, as we will see.)

The basic proof strategy for Theorem~\ref{thm:greedy-main} is similar to that of maximizing submodular functions, namely showing that in every iteration, the value of $f(\cdot)$ increases significantly. The key lemma is the following.

\begin{lemma} \label{lem:large-gain}
Let $S, T$ be two sets of columns, with $f_A(S) \ge f_A(T)$.  Then there exists $v \in S$ such that
%For $r \ge 1$, let $T_r$ denote the set of columns output by $\greedy(A, r)$. Then for any set $S$ of vectors such that $f_A(S) \geq f_A(T_r)$, there exists a 
\[ f_A(T \cup v) - f_A(T) \ge \sigma_{\min}(S) \frac{\big(f_A(S) - f_A(T)\big)^2}{4|S|f_A(S)}.\]
\end{lemma}

Theorem~\ref{thm:greedy-main} follows easily from Lemma~\ref{lem:large-gain}, which we show at the end of the section. Thus let us first focus on proving the lemma.  Note that for submodular $f$, the analogous lemma simply has $\frac{f(S) - f(T)}{|S|}$ on the right-hand side (RHS).
The main ingredient in the proof of Lemma~\ref{lem:large-gain} is its {\em single vector} version:
\begin{lemma}\label{lem:one-vector}
Let $S, T$ be two sets of columns, with $f_u(S) \ge f_u(T)$. Suppose $S=\{v_1, \dots, v_k\}$.  Then
\[ \sum_{i=1}^k \Big( f_u(T \cup v_i) - f_u(T) \Big) \ge \sigma_{\min}(S) \frac{\big(f_u(S) - f_u(T)\big)^2}{4f_u(S)}.\]
\end{lemma}

Let us first see why Lemma~\ref{lem:one-vector} implies Lemma~\ref{lem:large-gain}. Observe that for any set of columns $T$, $f_A (T) = \sum_{j} f_{A_j} (T)$ (sum over the columns), by definition. For a column $j$, let us define $\delta_j = \min \{ 1, \frac{f_{A_j}(T)}{f_{A_j}(S)}\}$. Now, using Lemma~\ref{lem:one-vector} and plugging in the definition of $\delta_j$, we have
\begin{align}
  & \frac{1}{\sigma_{\min}(S)} \sum_{i=1}^k  \big( f_A(T \cup v_i) - f_A(T) \big) \label{eq:start}\\
  & \quad =  \frac{1}{\sigma_{\min}(S)} \sum_{j = 1}^n \sum_{i=1}^k \big( f_{A_j}(T \cup v_i) - f_{A_j}(T) \big) \notag\\
& \quad \geq \sum_{j=1}^n \frac{ (1-\delta_j)^2 f_{A_j}(S)}{4} \label{eq:temp3}\\
& \quad = \frac{f_A(S)}{4} \sum_{j=1}^n (1 - \delta_j)^2 \frac{f_{A_j}(S)}{f_A(S)} \label{eq:temp4}\\ 
& \quad \geq \frac{f_A(S)}{4} \left( \sum_{j=1}^n (1 - \delta_j) \frac{f_{A_j}(S)}{f_A(S)}\right)^2 \label{eq:temp5} \\ 
& \quad = \frac{1}{4 f_A(S)} \Big(\sum_{j=1}^n \max\{ 0, f_{A_j}(S) - f_{A_j}(T) \}\Big)^2 \label{eq:temp6} \\
%\\ &\ge \frac{1}{4 f_A(S)} \Big(\sum_{j=1}^n f_{A_j}(S) - f_{A_j}(T) \Big)^2 \label{temp eq 7}
& \quad \ge \frac{1}{4 f_A(S)} \Big(f_A(S) - f_A(T) \Big)^2 \label{eq:temp8}
\end{align}
To get \eqref{eq:temp5}, we used Jensen's inequality ($\mathbb{E}[X^2] \geq ( \mathbb{E}[X])^2$) treating $\frac{f_{A_j}(S)}{f_{A}(S)}$ as a probability distribution over indices $j$. Thus it follows that there exists an index $i$ for which the gain is at least a $\frac{1}{|S|}$ factor, proving Lemma~\ref{lem:large-gain}.
%Let us thus focus on showing Lemma~\ref{lem:one-vector}. 

\begin{proof}[Proof of Lemma~\ref{lem:one-vector}]
Let us first analyze the quantity $f_u(T \cup v_i) - f_u(T)$, for some $v_i \in S$.  As mentioned earlier, we may assume the $v_i$ are normalized. If $v_i \in \Span(T)$, this quantity is $0$. Thus we can assume that such $v_i$ have been removed from $S$. Now, adding $v_i$ to $T$ gives a gain because of the component of $v_i$ orthogonal to $T$, i.e., $v_i - \Pi_T v_i$, where $\Pi_T$ denotes the projector onto $\Span(T)$. Define
\[ v_i' = \frac{v_i - \Pi_T v_i}{\norm{v_i - \Pi_T v_i}}_2.\] By definition, $\Span(T \cup v_i) = \Span(T \cup v_i')$. Thus the projection of a vector $u$ onto $\Span(T \cup v_i')$ is $\Pi_T u + \iprod{u, v_i'} v_i'$, which is a vector whose squared length is
$\norm{\Pi_T u}^2 + \iprod{u, v_i'}^2 = f_u(T) + \iprod{u, v_i'}^2$.
This implies that
\begin{equation}\label{eq:gain-single}
f_u(T \cup v_i ) - f_u(T) = \iprod{u, v_i'}^2.
\end{equation}

Thus, to show the lemma, we need a lower bound on $\sum_i \iprod{u, v_i'}^2$. Let us start by observing
that a more explicit definition of $f_u(S)$ is the squared-length of the projection of $u$ onto $\Span(S)$, i.e. $f_u(S) = \max_{x \in \Span(S), \norm{x}_2 = 1} \iprod{u, x}^2$. Let $x = \sum_{i=1}^k \alpha_i v_i$ be a maximizer. Since $\norm{x}_2=1$, by the definition of the smallest squared singular value, we have $\sum_i \alpha_i^2 \le \frac{1}{\sigma_{\min}(S)}$.  Now, decomposing $x = \Pi_T x + x'$, we have
\[ f_u(S) =  \langle x, u \rangle^2  
= \langle x' + \Pi_T x, \; u \rangle^2 
= (\langle x', u \rangle + \langle \Pi_T x, u \rangle)^2.\]
Thus (since the worst case is when all signs align),
\begin{align}
|\iprod{ x', u}| &\ge \sqrt{f_u(S)} - |\langle \Pi_T x, u \rangle| %\label{eq:dotprod-lb}
 \ge \sqrt{f_u(S)} - \sqrt{f_u(T)} \notag \\ %  \\
&= \frac{f_u(S) - f_u(T)}{\sqrt{f_u(S)}+ \sqrt{f_u(T)}} \ge \frac{f_u(S) - f_u(T)}{2\sqrt{f_u(S)}}. \label{eq:dotprod-lb2}
\end{align}
where we have used the fact that $|\iprod{\Pi_T x, u}|^2 \le f_u(T)$, which is true from the definition of $f_u(T)$ (and since $\Pi_T x$ is a vector of length $\le 1$ in $\Span(T)$). 

Now, because $x = \sum_i \alpha_i v_i$, we have $x' = x - \Pi_T x = \sum_i \alpha_i(v_i - \Pi_T v_i) = \sum_i \alpha_i\norm{v_i - \Pi_Tv_i}_2v_i'$. Thus,
\begin{align*}
\iprod{x', u}^2 &= 
\big( \sum_i \alpha_i \norm{v_i - \Pi_Tv_i}_2 \iprod{v_i' , u} \big)^2
\\ &\le \big( \sum_i \alpha_i^2 \norm{v_i - \Pi_Tv_i}_2^2 \big) \big( \sum_i \iprod{v_i', u}^2 \big)
\\ &\le \big( \sum_i \alpha_i^2 \big) \big( \sum_i \iprod{v_i', u}^2 \big).
\end{align*}
%\[ 
%(\iprod{x', u})^2 = 
%\big( \sum_i \alpha_i \iprod{v_i' , u} \big)^2 \le \big( \sum_i \alpha_i^2 \big) \big( \sum_i \iprod{v_i', u}^2 \big). \]
where we have used Cauchy-Schwartz, and then the fact that $\|v_i - \Pi_Tv_i\|_2 \leq 1$ (because $v_i$ are unit vectors). Finally, we know that $\sum_i \alpha_i^2 \le \frac{1}{\sigma_{\min}(S)}$, which implies 
\[ \sum_i \iprod{v_i', u}^2 \ge \sigma_{\min}(S)  \iprod{x', u}^2 \ge \sigma_{\min}(S) \frac{(f_u(S)- f_u(T))^2}{4f_u(S)}\]
Combined with~\eqref{eq:gain-single}, this proves the lemma. 
\end{proof}
%
%We can now complete the proof of Theorem~\ref{thm:greedy-main}.
%\subsection{Proof of Theorem~\ref{thm:greedy-main}}\label{sec:lemma-implies-greedy-thm}
%We will only outline the proof (details in Section~\ref{app:greedy-main} of the Appendix). Let $T_r$ denote the set of columns picked after $r$ steps of $\greedy$, and suppose $f_A(T_r) = (1-\gamma_r) f_A(OPT_k)$. Let us also denote $\Delta = 1/\sigma_{\min}(OPT_k)$, for convenience. Then, using Lemma~\ref{lem:large-gain} (which implies that there is a column we can add, and $\greedy$ will only do better), we have that 
%\[ f_A(T_{r+1}) - f_A(T_r) \ge \frac{ \gamma_r^2 f_A(OPT_k)}{ \Delta k}, \]
%or equivalently, $\gamma_{r+1} \le \gamma_r - \gamma_r^2/(\Delta k)$.  This means that after roughly $\frac{\Delta k}{\gamma_r}$ steps, $\gamma_r$ reduces by a factor $1/2$.  We start with $\gamma_r =1$. To go to $\epsilon$, we need (if $\epsilon = 2^{-q}$)
%\[ \Delta k \left( 1 + 2 + \dots + 2^{q} \right) \approx \frac{2 \Delta k}{\epsilon} \]  
%steps. This completes the proof.
% \qed
%\subsection{Proof of Theorem~\ref{thm:greedy-main}}\label{sec:lemma-implies-greedy-thm}

 \begin{proof}[Proof of Theorem~\ref{thm:greedy-main}]
 For notational convenience, let $\sigma = \sigma_{min}(OPT_k)$ and $F = f_A(OPT_k)$. Define $\Delta_0 = F$, $\Delta_1 = \frac{\Delta_0}{2}$, $\dots$, $\Delta_{i+1} = \frac{\Delta_i}{2}$ until $\Delta_N \leq \eps F$. Note that the gap $f_A(OPT_k) - f_A(T_0) = \Delta_0$. We show that it takes at most $\frac{8kF}{\sigma \Delta_i}$ iterations (i.e. additional columns selected) to reduce the gap from $\Delta_i$ to $\frac{\Delta_i}{2} = \Delta_{i+1}$. To prove this, we invoke Lemma \ref{lem:large-gain} to see that the gap filled by $\frac{8kF}{\sigma \Delta_i}$ iterations is at least $\frac{8kF}{\sigma \Delta_i} \cdot \sigma \frac{(\frac{\Delta_i}{2})^2}{4kF}
 = \frac{\Delta_i}{2} = \Delta_{i+1}$. Thus the total number of iterations $r$ required to get a gap of at most $\Delta_N \leq \eps F$ is:
\vspace{-0.1cm}
\[
 r
\leq \sum_{i=0}^{N-1} \frac{8kF}{\sigma \Delta_i}
= \frac{8kF}{\sigma} \sum_{i=0}^{N-1} \frac{2^{i-N+1}}{\Delta_{N-1}} 
% \label{Large marginal gain lemma eq 2}
 < \frac{16k}{\eps \sigma}
.\]
 where the last step is due to $\Delta_{N-1} > \eps F$ and $\sum_{i=0}^{N-1} 2^{i-N+1} < 2$. Therefore, after $r < \frac{16k}{\eps \sigma}$ iterations, we have $f_A(OPT_k) - f_A(T_r) \leq \eps f_A(OPT_k)$. Rearranging proves the lemma.
 \end{proof}

\section{Distributed Greedy Algorithm} \label{section-4}
We will now analyze the distributed version of the greedy algorithm that was discussed earlier. We show that in one {\em round}, we will find a set of size $O(k)$ as before, that has an objective value $\Omega(f(\opt_k)/\kappa)$, where $\kappa$ is a condition number (defined below). We also combine this with our earlier ideas to say that if we perform $O(\kappa / \eps)$ {\em rounds} of \distgreedy, we get a $(1-\eps)$ approximation (Theorem~\ref{thm:core-set-2}).

\subsection{Analyzing one round}

We consider an instance of $\gcss(A, B, k)$, and let $\opt$ denote an optimum set of $k$ columns.
 Let $\ell$ denote the number of machines available.
 The columns (of $B$) are partitioned across machines, such that machine $i$ is given columns $T_i$. It runs $\greedy$ as explained earlier and outputs $S_i \subset T_i$ of size $k' = \frac{32k}{\sigma_{\min}(OPT)}$. Finally, all the $S_i$ are moved to one machine and we run $\greedy$ on their union and output a set $S$ of size $k'' = \frac{12k}{\sigma_{\min}(OPT)}$. %  containing $k'$ of them. For all $i \in [m]$, define:
Let us define $\kappa(\opt) = \frac{\sigma_{\max}(\opt)}{\sigma_{\min}(\opt)}$.

\begin{thm}\label{thm:distributed-main}
Consider running $\distgreedy$ on an instance of $\gcss(A, B, k)$. We have
\[ \E[ \max\{ f_A(S), \max_i \{f_A(S_i)\}\}] \ge \frac{f(\opt)}{8\cdot \kappa(\opt)}.\] 
\end{thm}
The key to our proof are the following definitions:
\begin{align*}
\opt_i^S &= \{x \in \opt \; : \; x \in \greedy(A, T_i \cup x, k') \}\\ 
\opt_i^{NS} &= \{x \in \opt \; : \; x \not \in \greedy(A, T_i \cup x, k') \}
\end{align*}
In other words, $OPT_i^S$ contains all the vectors in $OPT$ that would have been selected by machine $i$ if they had been added to the input set $T_i$. By definition, the sets $(OPT_i^S, OPT_i^{NS})$ form a partition of $OPT$ for every $i$.

\iffalse
\begin{remark} \label{rem-opt-partition}
By definition, $OPT = OPT_i^S \cup OPT_i^{NS}$ is a union of disjoint sets for all $i \in [m]$. That is, $OPT_i^S$ and $OPT_i^{NS}$ partition $OPT$.
\end{remark}
\fi

{\bf Proof outline.} Consider any partitioning $T_1, \dots, T_\ell$, and consider the sets $\opt_i^{NS}$. Suppose one of them (say the $i$th) had a large value of $f_A(\opt_i^{NS})$. Then, we claim that $f_A(S_i)$ is also large. The reason is that the greedy algorithm does {\em not} choose to pick the elements of $\opt_i^{NS}$ (by definition) -- this can only happen if it ended up picking vectors that are ``at least as good''. This is made formal in Lemma~\ref{lem:opt-ns}.  Thus, we can restrict to the case when {\em none} of $f_A(\opt_i^{NS})$ is large. In this case, Lemma~\ref{lem:additivity} shows that $f_A(\opt_i^{S})$ needs to be large for each $i$. Intuitively, it means that most of the vectors in $\opt$ will, in fact, be picked by $\greedy$ (on the corresponding machines), and will be considered when computing $S$. The caveat is that we might be unlucky, and for every $x \in \opt$, it might have happened that it was sent to machine $j$ for which it was not part of $\opt_j^{S}$. We show that this happens with low probability, and this is where the random partitioning is crucial (Lemma~\ref{lem:opt-s}).  This implies that either $S$, or one of the $S_i$ has a large value of $f_A(\cdot)$.

Let us now state two lemmas, and defer their proofs to Sections~\ref{app:opt-ns} and~\ref{app:additivity} respectively. % turn to the formal statements and proofs.

\begin{lemma} \label{lem:opt-ns}
For $S_i$ of size $k' = \frac{32 k}{\sigma_{\min}(OPT)}$, we have
\[ f(S_i) \geq \frac{f_A(OPT_i^{NS})}{2} ~\text{ for all $i$.}\] 
%$f(S_i) \geq \frac{f_A(OPT_i^{NS})}{2}$ for all $i$.
\end{lemma}

\begin{lemma} \label{lem:additivity}
For any matrix $A$, and any partition $(I, J)$ of $OPT$:
\begin{equation}\label{eq:to-prove-additivity}
f_A(I) + f_A(J) \geq \frac{f_A(\opt)}{2\kappa(\opt)}.
\end{equation}
\end{lemma}
Our final lemma is relevant when none of $f_A(\opt_i^{NS})$ are large and, thus, $f_A(\opt_i^{S})$ is large for {\em all} $i$ (due to Lemma~\ref{lem:additivity}). In this case, Lemma~\ref{lem:opt-s} will imply that the expected value of $f(S)$ is large.

Note that $T_i$ is a random partition, so the $T_i$, the $\opt_i^{S}$, $\opt_i^{NS}$, $S_i$, and $S$ are all random variables. However, all of these value are fixed given a partition $\{T_i\}$. In what follows, we will write $f(\cdot)$ to mean $f_A(\cdot)$.

\begin{lemma}\label{lem:opt-s}
For a random partitioning $\{T_i\}$, and $S$ of size $k'' = \frac{12 k}{\sigma_{\min}(OPT)}$, we have
\begin{equation}\label{eq:main-lem-to-show}
\E[f(S)] \ge \frac{1}{2}\E \left[ \frac{\sum_{i=1}^{\ell} f(OPT_i^S)}{\ell}\right].
\end{equation}
\end{lemma}
\begin{proof}
At a high level, the intuition behind the analysis is that many of the vectors in $\opt$ are selected in the first phase, i.e., in $\cup_i S_i$. For an  $x \in \opt$, let $I_x$ denote the indicator for $x \in \cup_i S_i$.

Suppose we have a partition $\{T_i\}$. Then if $x$ had gone to a machine $i$ for which $x \in \opt_i^{S}$, then by definition, $x$ will be in $S_i$. Now the key is to observe (see definitions) that the event $x \in \opt_i^S$ does not depend on where $x$ is in the partition! In particular, we could think of partitioning all the elements except $x$ (and at this point, we know if $x \in \opt_i^S$ for all $i$), and {\em then} randomly place $x$.  Thus
\begin{equation}\label{eq:n14}
  \E[ I_x ] = \E \left[ \frac{1}{\ell} \sum_{i=1}^{\ell} [[x \in \opt_i^S]] \right],
\end{equation}
where $[[~\cdot~]]$ denotes the indicator.

We now use this observation to analyze $f(S)$. Consider the execution of the greedy algorithm on $\cup_i S_i$, and suppose $V^t$ denotes the set of vectors picked at the $t$th step (so $V^t$ has $t$ vectors). The main idea is to give a lower bound on% condition on $V^t$ and then obtain a lower bound on
\begin{equation}
\label{eq:diff-expectations}
 \E[ f(V^{t+1}) - f(V^t) ],
\end{equation}
where the expectation is over the partitioning $\{T_i\}$. Let us denote by $Q$ the RHS of \eqref{eq:main-lem-to-show}, for convenience.  Now, the trick is to show that for {\em any} $V^t$ such that $f(V^t) \le Q$, the expectation in~\eqref{eq:diff-expectations} is large. One lower bound on $f(V^{t+1}) - f(V^t)$ is (where $I_x$ is the indicator as above) 
\[\frac{1}{k} \sum_{x\in \opt} I_x \big( f(V^t \cup x) - f(V^t) \big).\]
Now for every $V$, we can use~\eqref{eq:n14} to obtain
\begin{align}
\E[ &f(V^{t+1}) - f(V^t) | V^t = V] \notag \\
&\ge \frac{1}{k\ell} \!\! \sum_{x \in\opt} \!\!\!\! \E \left[ \sum_{i=1}^\ell [[x \in \opt_i^S]]\right] \big( f(V \cup x) \!-\! f(V) \big) \notag \\
&= \frac{1}{k\ell} \E \left[ \sum_{i=1}^\ell \sum_{x \in \opt_i^S} \big( f(V \cup x) - f(V) \big) \right] \,. \notag
\end{align} 
  Now, using~\eqref{eq:start}-\eqref{eq:temp6}, we can bound the inner sum by
\[ \sigma_{\min}(\opt_i^S) \frac{ ( \max\{ 0, f(\opt_i^S) - f(V)\} )^2}{4f(\opt_i^S)} \,.\]
Now, we use $\sigma_{\min}(\opt_i^S) \ge \sigma_{\min}(\opt)$ and the identity that for any two nonnegative reals $a, b$: $(\max\{ 0, a-b \})^2/a \ge a/2 - 2b/3$.
%\[\frac{ (\max\{ 0, a-b \})^2}{a} \ge \frac{a}{2} - \frac{2b}{3} \,.\]
%(If $a < b$, the RHS is negative, so there is nothing to prove. Else, we can verify by expanding.) 
Together, these imply
\begin{align*}
\E[ & f(V^{t+1}) - f(V^t) | V^t = V ] \\
&\ge \frac{\sigma_{\min}(\opt)}{4k\ell} \E\left[ \sum_{i=1}^\ell \frac{f(\opt_i^S)}{2} - \frac{2 f(V)}{3} \right].
\end{align*}
and consequently: 
%\begin{align*}
$\E[  f(V^{t+1}) - f(V^t)] 
\ge  \alpha (Q - \frac{2}{3}\E[f(V^t)]$ 
%\end{align*}
for $\alpha = \sigma_{\min}(\opt)/4k$. 
If for some $t$, we have $\E[f(V^t)] \geq Q$, the proof is complete because $f$ is monotone, and $V^t \subseteq S$. 
Otherwise,  $\E[f(V^{t+1}) - f(V^t)]$ is at least $\alpha Q/3$ for each of the $k'' = 12 k/\sigma_{\min}(OPT) = 3/\alpha$ values of $t$. We conclude that $\E[f(S)]$ should be at least $(\alpha Q/3) \times (3/\alpha) = Q$ which completes the proof. 
\end{proof}
%
%Theorem~\ref{thm:distributed-main} now follows easily.
\begin{proof}[Proof of Theorem~\ref{thm:distributed-main}]
If $f_A( \opt_i^{NS}) \ge \frac{ f(\opt)}{4\kappa(\opt)}$ for some $i$, then we are done, because Lemma~\ref{lem:opt-ns} implies that $f_A(S_i)$ is large enough.  Otherwise, by Lemma~\ref{lem:additivity}, $f_A(\opt_i^{S}) \ge \frac{ f(\opt)}{4\kappa(\opt)}$ for all $i$. 
Now we can use Lemma~\ref{lem:opt-s} to conclude that $\E[ f_A(S) ] \ge \frac{ f(\opt)}{8\kappa(\opt)}$, completing the proof.
\end{proof}

\subsection{Multi-round algorithm}
We now show that repeating the above algorithm helps achieve a $(1-\eps)$-factor approximation.

We propose a framework with $r$ epochs for some integer $r>0$. In each epoch $t \in [r]$, we run the $\distgreedy$ algorithm to select set $S^t$. The only thing that changes in different epochs is the objective function: in epoch $t$, the algorithm selects columns based on the function $f^t$ which is defined to be: $f^t(V) = f_A(V \cup S^1 \cup S^2 \cdots \cup S^{t-1})$ for any $t$. We note that function $f^1$ is indeed the same as $f_A$. The final solution is the union of solutions: $\cup_{t=1}^r S^t$.

\begin{thm}\label{thm:core-set-2}
For any $\eps <1$, the expected value of the solution of the $r$-epoch $\distgreedy$ algorithm, for $r = O(\kappa(\opt)/\epsilon)$, is at least $(1-\eps)f(OPT)$.
%For any $r \geq 1$, the expected value of solution of multi-epoch greedy algorithm is at least $[1 - (1 - \Omega(\sigma_{min}(OPT)))^r] f(OPT)$. In particular, with $r = O(1/\sigma_{min}(OPT))$ rounds, the expected approximation factor becomes at least any constant (say $0.9$), and with $O(log(1/\epsilon)/\sigma_{min}(OPT))$ rounds the expected value of the solution is at least $(1 - \epsilon)f(OPT)$. 
\end{thm}

The proof is provided in Section~\ref{app:core-set-2} of the appendix.

{\em Necessity of Random Partitioning.}  We point out that the random partitioning step of our algorithm is crucial for the $\gcss(A, B, k)$ problem.  We adapt the instance from~\cite{VahabPODS2014} and show that even if each machine can compute $f_A(\cdot)$ exactly, and is allowed to output $\text{poly}(k)$ columns, it cannot compete with the optimum. Intuitively, this is because the partition of the columns in $B$ could ensure that in each partition $i$, the best way of covering $A$ involve picking some vectors $S_i$, but the $S_i$'s for different $i$ could overlap heavily, while the global optimum should use different $i$ to capture different {\em parts} of the space to be covered. (See Theorem \ref{thm:rand-part} in Appendix~\ref{app:rand-part} for details.)

\section{Further optimizations for \greedy} \label{section-5}
We now elaborate on some of the techniques discussed in Section~\ref{section-2} for improving the running time of $\greedy$. 
We first assume that we left-multiply both $A$ and $B$ by a random Gaussian matrix of dimension $r \times m$, for $r \approx k \log n/\eps^2$. Working with the new instance suffices for the purposes of $(1-\eps)$ approximation to CSS (for picking $O(k)$ columns). (Details in the Appendix, Section~\ref{app:random-projections})

\subsection{Projection-Cost Preserving Sketches}
Marginal gain evaluations of the form $f_A(S \cup v) - f_A(S)$ require summing the marginal gain of $v$ onto each column of $A$. When $A$ has a large number of columns, this can be very expensive. To deal with this, we use a {\em sketch} of $A$ instead of $A$ itself. This idea has been explored in several recent works; we use the following notation and result: % recent work of \cite{Cohen} for concreteness. %The following is definition 1 from their paper, presented in our notation for consistency.

\begin{defin}[\cite{Cohen}] \label{defin:pcps} For a matrix $A \in \Real^{m \times n}$, 
$A' \in \Real^{m \times n'}$ is a \emph{rank-$k$ Projection-Cost Preserving Sketch (PCPS)} with error $0 \leq \eps < 1$ if for any set of $k$ vectors $S$, we have: $(1 - \eps) f_A(S) \leq f_{A'}(S) + c \leq (1 + \eps) f_A(S)$
%\[ (1 - \eps) f_A(S) \leq f_{A'}(S) + c \leq (1 + \eps) f_A(S), \]
where $c \geq 0$ is a constant that may depend on $A$ and $A'$ but is independent of $S$.
\end{defin}

\begin{thm}\label{thm:pcps}[Theorem 12 of \cite{Cohen}]
Let $R$ be a random matrix with $n$ rows and $n' = O(\frac{k + \log\frac{1}{\delta}}{\eps^2})$ columns, where each entry is set independently and uniformly to $\plusminus \sqrt{\frac{1}{n'}}$. Then for any matrix $A \in \Real^{m \times n}$, with probability at least $1 - O(\delta)$, $AR$ is a rank-$k$ PCPS for $A$.
\end{thm}

Thus, we can use PCPS to sketch the matrix $A$ to have roughly $k/\eps^2$ columns, and use it to compute $f_A(S)$ to a $(1\pm \eps)$ accuracy for any $S$ of size $\le k$.  This is also used in our distributed algorithm, where we send the sketch to every machine.

\subsection{Lazier-than-lazy Greedy}
The natural implementation of $\greedy$ requires $O(nk)$ evaluations of $f(\cdot)$ since we compute the marginal gain of all $n$ candidate columns in each of the $k$ iterations. For submodular functions, one can do better: the recently proposed $\ltlgreedy$ algorithm obtains a  $1 - \frac{1}{e} - \delta$ approximation with only a linear number $O(n \log (1/\delta))$ of marginal gain evaluations \cite{Mirzasoleiman}. We show that a similar result holds for \gcss, even though our cost function $f(\cdot)$ is not submodular. 

The idea is as follows. Let $T$ be the current solution set. To find the next element to add to $T$, we draw a sized $\frac{n_B \log (1/\delta)}{k} $ subset uniformly at random from the columns in $B \setminus T$. We then take from this set the column with largest marginal gain, add it to $T$, and repeat. We show this gives the following guarantee (details in Appendix Section~\ref{app:thm-lazier-than-lazy}.)

\begin{thm} \label{thm-lazier-than-lazy:main}
Let $A\in \Real^{m \times n_A}$ and $B \in \Real^{m \times n_B}$. Let $OPT_k$ be the set of columns from $B$ that maximizes $f_A(S)$ subject to $|S| = k$. Let $\eps, \delta > 0$ be any constants such that $\epsilon + \delta \leq 1$. Let $T_r$ be the set of columns output by $\ltlgreedy (A, B, r)$, for $r = \frac{16k}{\eps \sigma_{\min}(OPT_k)}$. Then we have:
$$\E[f_A(T_r)] \geq (1 - \eps - \delta) f_A(OPT_k)$$
Further, this algorithm evaluates marginal gain only a linear number $\frac{16n_B \log (1/\delta)}{\eps \sigma_{\min}(OPT_k)}$ of times.
\end{thm}

Note that this guarantee is nearly identical to our analysis of $\greedy$ in Theorem \ref{thm:greedy-main}, except that it is in expectation. % as opposed to worst case, and loses a small factor of $\frac{1 - \eps - \delta}{1 - \eps}$.
The proof strategy is very similar to that of Theorem \ref{thm:greedy-main}, namely showing that the value of $f(\cdot)$ increases significantly in every iteration (see Appendix Section~\ref{app:thm-lazier-than-lazy}).

{\bf Calculating marginal gain faster.} We defer the discussion to Appendix Section~\ref{sec:app:marginal}. % because of limited space.

\section{Experimental results}\label{section-6}

\begin{figure*}[t]
\centering
\includegraphics[scale=0.32]{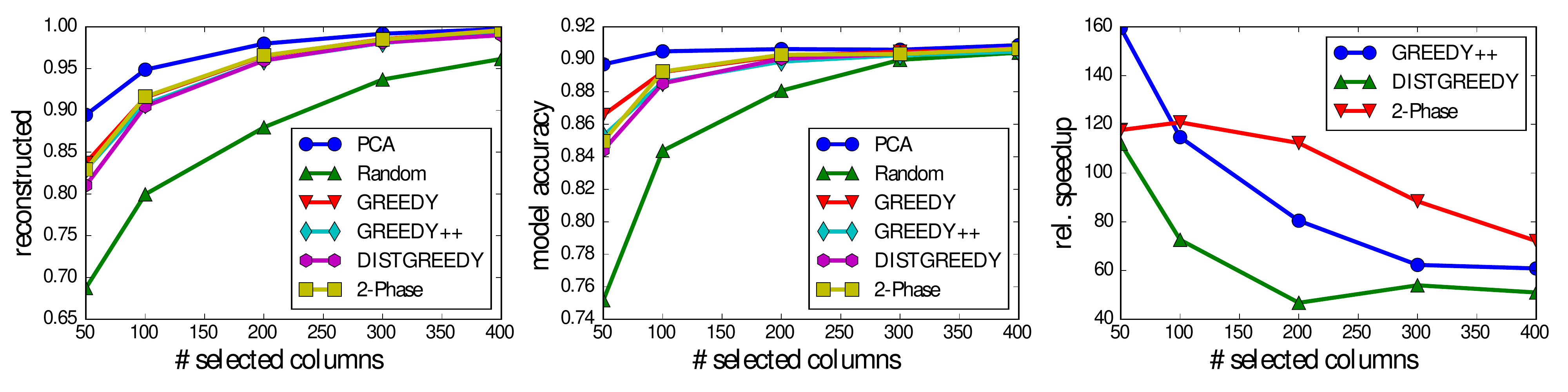}
%  \vspace{-0.3cm}
\caption{A comparison of reconstruction accuracy, model classification
accuracy and runtime of various column selection methods (with PCA
proved as an upper bound). The runtime is shown plot shows the
relative speedup over the naive GREEDY algorithm.}
\label{fig}
\end{figure*}

In this section we present an % preliminary
empirical investigation of
the GREEDY, GREEDY++ and \distgreedy\ algorithms.
Additionally, we will compare with several baselines:
%\vspace{-0.2cm}
%\begin{description}
  %\item[Random]
    \\
    {\bf Random:}
    The simplest imaginable baseline, this method selects
    columns randomly. \\
  %\item[Leverage]
    {\bf 2-Phase:}
    The two-phased algorithm of \cite{Boutsidis2}, which
operates by first sampling $\Theta(k \log k)$ columns based on
properties of the top-$k$ right singular space of the input matrix
(this requires computing a top-$k$ SVD), then finally selects exactly
$k$ columns via a deterministic procedure. The overall complexity is
dominated by the top-$k$ SVD, which is $O( \min\{mn^2, m^2n\})$. \\
  %\item[PCA]
    {\bf PCA:}
    The columns of the rank-$k$ PCA projection matrix will be
used to serve as an upper bound on performance, as they explicitly
minimize the Forbenius reconstruction criteria. Note this method only
serves as an upper bound and does not fall into  the framework of
column subset selection.
%\end{description}

We investigate using these algorithms using two datasets, one with a small
set of columns (mnist) that is used to compare both scalable and
non-scalable methods, as well as a sparse dataset with a large number of
columns (news20.binary) that is meant to demonstrate the scalability of the
GREEDY core-set algorithm.\footnote{Both datasets can be
found at:
www.csie.ntu.edu.tw/$\sim$cjlin/libsvmtools/datasets/multiclass.html.}

Finally, we are also interested in the effect of column selection as a
preprocessing step for supervised learning methods. To that end, we will train
a linear SVM model, using the LIBLINEAR library \citep{fan2008}, with the
subselected columns (features) and measure the effectiveness of the model on a
held out test set. For both datasets we report test error for the best
choice of regularization parameter $c \in \{10^{-3}, \ldots, 10^4\}$.
We run GREEDY++ and \distgreedy\ with $\frac{n}{k}\log(10)$
marginal gain evaluations per iteration and the distributed
algorithm uses $s = \sqrt{\frac{n}{k}}$ machines with each machine
recieving $\frac{n}{s}$ columns.

\subsection{Small scale dataset (mnist)}
We first consider the MNIST digit recognition task, which is a
ten-class classification problem. There are $n$ = 784 input features
(columns) that represent pixel values from the $28\times28$-pixel
images. We use $m$ = 60,000 instances to train with and 10,000 instances
for our test set.

From Figure~\ref{fig} we see that all column sampling methods, apart
from Random, select columns that approximately provide the same amount
of reconstruction and are able to reach within 1\% of the performance
of PCA after sampling 300 columns. We also see a very similar trend
with respect to classification accuracy. It is notable that, in
practice, the core-set version of GREEDY incurs almost no additional
error (apart from at the smallest values of $k$) when compared to the
standard GREEDY algorithm.

Finally, we also show the relative speed up of the competitive methods
over the standard GREEDY algorithm. In this small dataset regime, we
see that the core-set algorithm does not offer an improvement over the single
machine GREEDY++ and in fact the 2-Phase algorithm is the fastest.
This is primarily due to the overhead of the distributed core-set
algorithm and the fact that it requires two greedy selection stages
(e.g.\ map and reduce). Next, we will consider a dataset that is large enough that a distributed model is in fact necessary. %In the next section, we will a setting where
%this distributed computation model is in fact necessary and where
%sophisticated single machine algorithms will not scale.

\subsection{Large scale dataset (news20.binary)}

\begin{table}
\begin{small}
\begin{center}
  \begin{tabular}{|c|c|c|c|c|}
    \hline
    \bf n & \bf Rand & \bf 2-Phase & \bf $\distgreedy$ &  \bf PCA \\
    \hline
%    \multicolumn{5}{c}{Accuracy (percentage)} \\
    \hline
    500 & 54.9 & 81.8 (1.0) & 80.2 (72.3) & 85.8 (1.3) \\
    \hline
    1000 & 59.2 & 84.4 (1.0) & 82.9 (16.4) & 88.6 (1.4)\\
    \hline
    2500 & 67.6 & 87.9 (1.0) & 85.5 (2.4) & 90.6 (1.7) \\
    \hline
%    \multicolumn{5}{c}{Selection Cost (seconds)} \\
%    \hline
%    500 & -- & 3180 & 44 & TBD \\
%    \hline
%    1000 & -- & 3480 & 212 & TBD \\
%    \hline
%    2500 & -- & 4713 & 1907 & TBD \\
%    \hline
  \end{tabular}
\end{center}
\vspace{-0.1cm}
\end{small}
  \caption{A comparison of the classification accuracy of selected features. Also, the relative speedup over the 2-Phase algorithm for selecting features is shown in parentheses.
  %Note, random selection essentially requires no computation beyond
  %several calls to a random number generator.
\vspace{-0.4cm}
  }
\label{table}
\end{table}

In this section, we show that the $\distgreedy$ algorithm can indeed
scale to a dataset with a large number of columns.
%We note that in
%this setting, methods that rely on decompositions (such as SVD) no
%longer work trivially and although scalable approximate methods are
%known, their use is beyond the scope of this investigation.
The news20.binary dataset is a binary class text classification
problem, where we start with $n$ = 100,000 sparse features (0.033\%
non-zero entries) that represent text trigrams, use $m$ = 14,996
examples to train with and hold-out 5,000 examples to test with.

We compare the classification accuracy and column selection runtime of
the naive random method, 2-Phase algorithm
%as well as to sampling proportional to the column $L_2$-norm
as well as PCA (that serves as an upper bound on performance)
to the $\distgreedy$ algorithm.
%
%We find for $k$ = 500 that features from random selection achieve a
%classification accuracy of 54.9\%, proportional random sampling
%achieve 76.6\% while \distgreedy\ achieves 81.6\%. Similarly for
%$k$=1000 we observe 59.2\% for random, 83.7\% for proportional random
%while our method achieve 85.4\% accuracy. 
%For both these setting, the running time is in the order of minutes. 
The results are presented in Table~\ref{table}, which shows that
$\distgreedy$ and 2-Phase both perform significantly better than random sampling and
come relatively close to the PCA upper bound in terms of accuracy. However,
we also find that $\distgreedy$ can be magnitudes of order faster than the 2-Phase
algorithm. This is in a large part because the 2-Phase algorithm suffers from the bottleneck of
computing a top-$k$ SVD. We note that an approximate SVD method
could be used instead,
however, it was outside the scope of this preliminary empirical investigation.

In conclusion, we have demonstrated that \distgreedy\ is able to scale to
larger sized datasets while still selecting effective features.

\bibliographystyle{icml2016}
\bibliography{icml_main}

\clearpage
\appendix

\section{Appendix}

\subsection{Proof of Lemma~\ref{lem:opt-ns}}\label{app:opt-ns}
\begin{proof}
Let us fix some machine $i$. The main observation is that running greedy with $T_i$ is the same as running it with $T_i \cup \opt_i^{NS}$ (because by definition, the added elements are not chosen).  Applying Theorem~\ref{thm:greedy-main}\footnote{To be precise, Theorem~\ref{thm:greedy-main} is presented as comparing against the optimum set of $k$ columns. However, an identical argument (simply stop at the last line in the proof of Theorem~\ref{thm:greedy-main}) shows the same bounds for any (potentially non-optimal) set of $k$ columns. This is the version we use here.} with $B = T_i \cup \opt_i^{NS}$ and $\eps = \frac{1}{2}$, we have that for
$ k' \ge \frac{32 |\opt_i^{NS}|}{\sigma_{\min}(\opt_i^{NS})}$,
then $f_A(S_i) \geq \frac{f_A(\opt_i^{NS})}{2}$. Now since $\opt_i^{NS}$ is a subset of $\opt$, we have that $\opt_i^{NS}$ is of size at most $|\opt| = k$, and also $\sigma_{\min}(\opt_i^{NS}) \ge \sigma_{\min}(\opt)$. Thus the above bound certainly holds whenever $k' \ge \frac{32 k}{\sigma_{\min}(\opt)}$.
\end{proof}

\subsection{Proof of Lemma~\ref{lem:additivity}} \label{app:additivity}

\begin{proof} [Lemma~\ref{lem:additivity}]
As before, we will first prove the inequality for one column $u$ instead of $A$, and adding over the columns gives the result. Suppose $\opt = \{v_1, \dots, v_k\}$, and let us abuse notation slightly, and use $I, J$ to also denote subsets of indices that they correspond to. Now, by the definition of $f$, there exists an $x = \sum_i \alpha_i v_i$, such that $\norm{x}=1$, and $\iprod{x, u}^2 = f_u(\opt)$.

Let us write $x = x_I + x_J$, where $x_I = \sum_{i \in I} \alpha_i v_i$. Then,
\begin{align*}
\iprod{x, u}^2 &= (\iprod{x_I, u} + \iprod{x_J, u})^2 \le 2(\iprod{x_I, u}^2 + \iprod{x_J, u}^2 ) \\
&\le 2(\norm{x_I}^2 f_u(I) + \norm{x_J}^2 f_u(J) ) \\
&\le 2(\norm{x_I}^2 + \norm{x_J}^2) (f_u(I) + f_u(J)).
\end{align*}
Now, we have
\[ \norm{x_I}^2 \le \sigma_{\max}(I)(\sum_{i \in I} \alpha_i^2),\]
from the definition of $\sigma_{\max}$, and we clearly have $\sigma_{\max}(I)\le \sigma_{\max}(\opt)$, as $I$ is a subset. Using the same argument with $J$, we have
\[ \norm{x_I}^2 + \norm{x_J}^2 \le \sigma_{\max}(\opt) (\sum_i \alpha_i^2). \]
Now, since $\norm{x}=1$, the definition of $\sigma_{\min}$ gives us that $\sum_i \alpha_i^2 \le 1/\sigma_{\min}(\opt)$, thus completing the proof.
\end{proof}

%\subsection{Proof of Theorem \ref{thm:greedy-main}} \label{app:greedy-main}
%\begin{proof}[Proof of Theorem~\ref{thm:greedy-main}]
%For notational convenience, let $F = f(OPT_k)$. Define $\Delta_0 = F$, $\Delta_1 = \frac{\Delta_0}{2}$, $\dots$, $\Delta_{i+1} = \frac{\Delta_i}{2}$ until $\Delta_N \leq \eps F$.
%\\ \\ Note that the gap $f(OPT_k) - f(T_0) \leq \Delta_0$. We show that it takes at most $\frac{8kF}{\sigma_{min}(OPT_k) \Delta_i}$ iterations (i.e. additional columns selected) to reduce the gap from $\Delta_i$ to $\frac{\Delta_i}{2} = \Delta_{i+1}$. To prove this, we invoke Lemma \ref{lem:large-gain} to see that the gap filled by $\frac{8kF}{\sigma_{min}(OPT_k) \Delta_i}$ iterations is at least $\frac{8kF}{\sigma_{min}(OPT_k) \Delta_i} \cdot \sigma_{min}(OPT_k) \frac{(\frac{\Delta_i}{2})^2}{4kF}
%= \frac{\Delta_i}{2} = \Delta_{i+1}$. Thus the total number of iterations $r$ required to get a gap of at most $\Delta_N \leq \eps F$ is:
%\begin{align}
%r
%&\leq \sum_{i=0}^{N-1} \frac{8kF}{\sigma_{min}(OPT_k) \Delta_i}
%\\ &= \frac{8kF}{\sigma_{min}(OPT_k)} \sum_{i=0}^{N-1} \frac{2^{i-N+1}}{\Delta_{N-1}} \label{Large marginal gain lemma eq 2}
%\\ &< \frac{16k}{\eps \sigma_{min}(OPT_k)} \label{Large marginal gain lemma eq 3}
%\end{align}
%where equation \eqref{Large marginal gain lemma eq 3} is due to $\Delta_{N-1} > \eps F$ and $\sum_{i=0}^{N-1} 2^{i-N+1} < 2$. Therefore, after $r < \frac{16k}{\eps \sigma_{min}(OPT_k)}$ iterations, we have $f_A(OPT_k) - f_A(T_r) \leq \eps f_A(OPT_k)$. Rearranging proves the lemma.
%\end{proof}

\subsection{Tight example for the bound in Theorem~\ref{thm:greedy-main}} \label{app:tight-ex}

We show an example in which we have a collection of (nearly unit)
vectors such that:
\enum{
\item Two of them can exactly represent a target vector
$u$ (i.e., $k=2$). 
\item The $\sigma_{\min}$ of the matrix with these two vectors as
columns is $\sim \theta^2$, for some parameter $\theta <1$.
\item The greedy algorithm, to achieve an error $\le \epsilon$ in
the squared-length norm, will require at least $\frac{1}{\theta^2
  \epsilon}$ steps.
}

The example also shows that using the greedy algorithm, we cannot
expect to obtain a multiplicative guarantee on the {\em error}. In the
example, the optimal error is zero, but as long as the full set of
vectors is not picked, the error of the algorithm will be non-zero.

\paragraph{The construction.} 
Suppose $e_0, e_1, \dots, e_n$ are orthogonal vectors.  The vectors in
our collection are the following: $e_1$, $\theta e_0 +e_1$, and
$2\theta e_0 + e_j$, for $j \ge 2$. Thus we have $n+1$ vectors. The
target vector $u$ is $e_0$.  Clearly we can write $e_0$ as a linear
combination of the first two vectors in our collection.

Let us now see what the greedy algorithm does. In the first step, it
picks the vector that has maximum squared-inner product with $e_0$.
This will be $2\theta e_0 + e_2$ (breaking ties arbitrarily). We claim
inductively that the algorithm never picks the first two vectors of
our collection. This is clear because $e_1, e_2, \dots, e_n$ are all
orthogonal, and the first two vectors have a strictly smaller
component along $e_0$, which is what matters for the greedy choice (it
is an easy calculation to make this argument formal).

Thus after $t$ iterations, we will have picked $2\theta e_0 + e_2,
2\theta e_0 + e_3, \dots, 2\theta e_0 + e_{t+1}$.  Let us call them
$v_1, v_2, \dots v_t$ resp. Now, what is the unit vector in the span
of these vectors that has the largest squared dot-product with $e_0$?

It is a simple calculation to find the best linear combination of the
$v_i$ -- all the coefficients need to be equal. Thus the best unit
vector is a normalized version of $(1/t) (v_1 + \dots v_t)$, which is
\[ v = \frac{ 2\theta e_0 + \frac{1}{t}(e_2 + e_3 +\dots e_{t+1})}{\sqrt{\frac{1}{t} + 4\theta^2}}. \]

For this $v$, to have $\iprod{u, v}^2 \ge 1-\epsilon$, we must have 
\[\frac{4\theta^2}{\frac{1}{t} + 4\theta^2} \le 1- \epsilon,  \]
which simplifies (for $\epsilon \le 1/2$) to $\frac{1}{4 t\theta^2} \le \frac{\epsilon}{2}$, or $t > \frac{1}{2\theta^2\epsilon}$.

\subsection{Proof of Theorem~\ref{thm-lazier-than-lazy:main}} \label{app:thm-lazier-than-lazy}
The key ingredient in our argument is that in every iteration, we obtain large marginal gain in expectation. This is formally stated in the following lemma.

\begin{lemma} \label{lem:large-gain-lazier-than-lazy}
Let $S, T$ be two sets of columns from $B$, with $f_A(S) \geq f_A(T)$. Let $|S| \leq k$, and let $R$ be a size $\frac{n_B \log \frac{1}{\delta}}{k}$ subset drawn uniformly at random from the columns of $B \setminus T$.  Then the expected gain in an iteration of $\ltlgreedy$ is at least $(1 - \delta) \sigma_{\min}(S) \frac{(f_A(S) - f_A(T))^2}{4kf_A(S)}$.
\end{lemma}

\begin{proof}[Proof of Lemma~\ref{lem:large-gain-lazier-than-lazy}]
The first part of the proof is nearly identical to the proof of Lemma 2 in \cite{Mirzasoleiman}. We repeat the details here for the sake of completeness.

Intuitively, we would like the random sample $R$ to include vectors we have not seen in $S \setminus T$. In order to lower bound the probability that $R \cap (S \setminus T) \neq \emptyset$, we first upper bound the probability that $R \cap (S \setminus T) = \emptyset$.
\begin{align}
\mathbb{P}\{R \cap (S \setminus T) = \emptyset\}
&= \Big(1 - \frac{|S \setminus T|}{n_B - |T|}\Big)^{\frac{n_B \log(\frac{1}{\delta})}{k}} \label{Stochastic greedy lemma eq 1}
\\ &\leq e^{-\frac{n_B \log(\frac{1}{\delta})}{k} \frac{|S \setminus T|}{n_B - |T|}} 
\\ &\leq e^{-\frac{\log(\frac{1}{\delta})}{k} |S \setminus T|} 
\end{align}
where we have used the fact that $1 - x \leq e^{-x}$ for $x \in \mathbb{R}$. Recalling that $\frac{|S\setminus T|}{k} \in [0, 1]$, we have:
\begin{align}
\mathbb{P}\{R \cap (S \setminus T) \neq \emptyset\}
&\geq 1 - e^{-\frac{\log(\frac{1}{\delta})}{k} |S \setminus T|}
\\ &\geq (1 - e^{-\log(\frac{1}{\delta})}) \frac{|S \setminus T|}{k}
\\ &= (1 - \delta) \frac{|S \setminus T|}{k} \label{Stochastic greedy lemma eq 2}
\end{align}
The next part of the proof relies on techniques developed in the proof of Theorem 1 in \cite{Mirzasoleiman}. For notational convenience, define $\Delta(v|T) = f_A(T \cup v) - f_A(T)$ to be the marginal gain of adding $v$ to $T$. Using the above calculations, we may lower bound the expected gain $\mathbb{E}[\max_{v \in R} \Delta(v | T)]$ in an iteration of $\ltlgreedy$ as follows:
\begin{align}
& \mathbb{E} \big[\max_{v \in R} \Delta(v | T)\big]
\\ &\geq (1 - \delta) \frac{|S \setminus T|}{k} \cdot \E[\max_{v \in R} \Delta(v | T) \; \Big| \; R \cap (S \setminus T) \neq \emptyset] \label{Stochastic greedy lemma eq 3}
\\ &\geq (1 - \delta) \frac{|S \setminus T|}{k} \cdot \E[\max_{v \in R \cap (S \setminus T)} \Delta(v | T) \; \Big| \; R \cap (S \setminus T) \neq \emptyset] \label{Stochastic greedy lemma eq 4}
\\ &\geq (1 - \delta) \frac{|S \setminus T|}{k} \cdot \frac{\sum_{v \in S \setminus T} \Delta(v | T)}{|S \setminus T|} \label{Stochastic greedy lemma eq 5}
\\ &\geq (1 - \delta) \sigma_{min}(S) \frac{(f_A(S)- f_A(T))^2}{4kf_A(S)} \label{Stochastic greedy lemma eq 6}
\end{align}
Equation \eqref{Stochastic greedy lemma eq 3} is due to conditioning on the event that $R \cap (S \setminus T) \neq \emptyset$, and lower bounding the probability that this happens with Equation \eqref{Stochastic greedy lemma eq 2}. Equation \eqref{Stochastic greedy lemma eq 5} is due to the fact each element of $S$ is equally likely to be in $R$, since $R$ is chosen uniformly at random. Equation \eqref{Stochastic greedy lemma eq 6} is a direct application of equation \eqref{eq:temp8} because $\sum_{v \in S \setminus T} \Delta(v | T) = \sum_{v \in S} \Delta(v | T)$.
\end{proof}

We are now ready to prove Theorem~\ref{thm-lazier-than-lazy:main}. The proof technique is similar to that of Theorem~\ref{thm:greedy-main}.
\begin{proof}[Proof of Theorem~\ref{thm-lazier-than-lazy:main}]
For each $i \in \{0, \dots, r\}$, let $T_i$ denote the set of $i$ columns output by $\ltlgreedy(A, B, i)$. We adopt the same notation for $F$ as in the proof of Theorem \ref{thm:greedy-main}. We also use a similar construction of $\{\Delta_0, \dots, \Delta_N\}$ except that we stop when $\Delta_N \leq \frac{\eps}{1 - \delta}F$.

We first demonstrate that it takes at most $\frac{8kF}{(1 - \delta) \sigma_{min}(OPT_k) \Delta_i}$ iterations to reduce the gap from $\Delta_i$ to $\frac{\Delta_i}{2} = \Delta_{i+1}$ in expectation. To prove this, we invoke Lemma \ref{lem:large-gain-lazier-than-lazy} on each $T_i$ to see that the expected gap filled by $\frac{8kF}{(1 - \delta) \sigma_{min}(OPT_k) \Delta_i}$ iterations is lower bounded by $\frac{8kF}{(1 - \delta) \sigma_{min}(OPT_k) \Delta_i} \cdot  (1 - \delta) \frac{\sigma_{min}(OPT_k)(\frac{\Delta_i}{2})^2}{4kF} = \frac{\Delta_i}{2} = \Delta_{i+1}$. Thus the total number of iterations $r$ required to decrease the gap to at most $\Delta_N \leq \frac{\eps}{1 - \delta} F$ in expectation is:
\begin{align}
r
&\leq \sum_{i=0}^{N-1} \frac{8kF}{(1 - \delta)\sigma_{min}(OPT_k) \Delta_i}
\\ &= \frac{8kF}{(1 - \delta)\sigma_{min}(OPT_k)} \sum_{i=0}^{N-1} \frac{2^{i-N+1}}{\Delta_{N-1}} \label{Stochastic greedy theorem eq 2}
\\ &< \frac{16k}{\eps \sigma_{min}(OPT_k)} \label{Stochastic greedy theorem eq 3}
\end{align}
where equation \eqref{Stochastic greedy theorem eq 3} is because $\Delta_{N-1} > \frac{\eps}{1 - \delta} F$ and $\sum_{i=0}^{N-1} 2^{i-N+1} < 1$. Therefore, after $r \geq \frac{16k}{\eps \sigma_{min}(OPT_k)}$ iterations, we have that:
\begin{align}
f_A(OPT_k) - \mathbb{E}[f_A(T_r)] 
&\leq \frac{\eps}{1 - \delta} f_A(OPT_k) 
\\ &\leq (\eps + \delta)f_A(OPT_k)
\end{align}
because $\eps + \delta \leq 1$. Rearranging proves the theorem.
\end{proof}

\subsection{Random Projections to reduce the number of rows} \label{app:random-projections}
Suppose we have a set of vectors $A_1, A_2, \dots, A_n$ in
$\re^m$, and let $\eps, \del$ be given accuracy parameters. For an integer $1 \le k \le n$, we say that a vector $x$ is in the $k$-span of $A_1, \dots, A_n$ if we can write $x = \sum_j \alpha_j A_j$, with at most $k$ of the $\alpha_j$ non-zero. Our main result of this section is the following.

\begin{thm} \label{thm:random-projections}
Let $1\le k \le n$ be given, and set $d = O^*(\frac{k \log
(\frac{n}{\delta \eps})}{\eps^2})$, where we use $O^*(\cdot)$ to omit $\log \log$ terms.  Let $G \in \re^{d \times m}$ be a matrix with entries drawn independently from $\calN(0,1)$. Then with probability at least $1-\delta$,
for {\em all} vectors $x$ that are in the $k$-span of
$A_1, A_2, \dots, A_n$, we have
\[ (1-\eps)\norm{x}_2 \le \frac{1}{\sqrt{d}} \norm{Gx}_2 \le (1+\eps) \norm{x}_2. \]
\end{thm}

The proof is a standard $\eps$-net argument that is similar to the proof of Lemma 10 in \cite{Sarlos}. Before giving the proof, we first state the celebrated lemma of Johnson and Lindenstrauss.
\begin{thm}\label{thm:jl-classic}\cite{Johnson}
Let $x \in \re^m$, and suppose $G \in \re^{d\times m}$ be a matrix with entries drawn independently from $\calN(0,1)$.  Then for any $\eps > 0$, we have \[ \Pr\left[ (1-\eps)\norm{x}_2 \le \frac{1}{\sqrt{d}} \norm{Gx}_2 \le (1+\eps) \norm{x}_2 \right] \ge 1- e^{-\eps^2d/4}. \]
\end{thm}

Now we prove Theorem ~\ref{thm:random-projections}.
\begin{proof}[Proof of Theorem ~\ref{thm:random-projections}]
The proof is a simple `net' argument for unit vectors in the $k$-span
of $A_1, \dots, A_n$.  The details are as follows.

First note that since the statement is scaling invariant, it suffices
to prove it for {\em unit} vectors $x$ in the $k$-span. Next, note
that it suffices to prove it for vectors in a $\gamma$-net for the
unit vectors in the $k$-span, for a small enough $\gamma$. To recall,
a $\gamma$-net for a set of vectors $S$ is a finite set subset
$\calN_\gamma$ with the property that for all $x \in S$, there exists a
$u \in \calN_\gamma$ such that $\norm{x - u}_2 \le \gamma$.

Suppose we fix some $\gamma$-net for the set of unit vectors in the
$k$-span of $A_1, \dots, A_n$, and suppose we show that for all $u \in
\calN_\gamma$, we have
\begin{equation}
  (1-\eps/2)\norm{u}_2 \le \frac{1}{\sqrt{d}} \norm{Gu}_2 \le (1+\eps/2) \norm{u}_2.\label{eq:eps-net-needed}
\end{equation}
Now consider any $x$ in the $k$-span. By definition, we can write $x =
u +w$, where $u \in \calN_\gamma$ and $\norm{w} <\gamma$.  Thus we
have
\begin{equation}
  \norm{Gu}_2 - \gamma \norm{G}_{2} \le \norm{Gx}_2 \le \norm{Gu}_2 + \gamma \norm{G}_2, \label{eq:eps-net-bound}
\end{equation}
where $\norm{G}_2$ is the spectral norm of $G$.  From now, let us set
$\gamma = \frac{\eps}{4\sqrt{d} \log(4/\delta)}$.  Now, whenever
$\norm{G}_2 < 2\sqrt{d}\log (4/\delta)$,
equation~\eqref{eq:eps-net-bound} implies
\[ \norm{Gu}_2 - \frac{\eps}{2} \le \norm{Gx}_2 \le \norm{Gu}_2 + \frac{\eps}{2}.\]

The proof follows from showing the following two statements: (a) there
exists a net $\calN_\gamma$ (for the above choice of $\gamma$) such
that Eq.~\eqref{eq:eps-net-needed} holds for all $u \in \calN_\gamma$
with probability $\ge 1- \delta/2$, and (b) we have $\norm{G}_2 <
2\sqrt{d}\log(4/\delta)$ w.p. at least $1-\delta/2$.

Once we have (a) and (b), the discussion above completes the proof. We
also note that (b) follows from the concentration inequalities on the
largest singular value of random matrices \cite{Rudelson}. Thus it only remains to prove (a).

For this, we use the well known result that every
$k$-dimensional space, the set of unit vectors in the space has a
$\gamma$-net (in $\ell_2$ norm, as above) of size at most
$(4/\gamma)^k$ \cite{Vershynin}. In our setting, there are $\binom{m}{k}$
such spaces we need to consider (i.e., the span of every possible
$k$-subset of $A_1, \dots, A_m$).  Thus there exists a $\gamma$-net
for the unit vectors in the $k$-span, which has a size at most
\[ \binom{m}{k} \cdot \left( \frac{4}{\gamma} \right)^k < \left( \frac{4m}{\gamma} \right)^k, \]
where we used the crude bound $\binom{m}{k} < m^k$.

Now Theorem~\ref{thm:jl-classic} implies that for any (given) $u \in
\calN_\gamma$ (replacing $\eps$ by $\eps/2$ and noting $\norm{u}_2 =
1$), that
\[ \Pr \left[  1+ \frac{\eps}{2} \le \frac{1}{\sqrt{d}} \norm{Gu}_2 \le 1+\frac{\eps}{2} \right] \ge 1- e^{-\eps^2 d/16}.\]

Thus by a union bound, the above holds for all $u \in \calN_\gamma$
with probability at least $1- |\calN_\gamma|e^{-\eps^2 d/16}$.
For our choice of $d$, it is easy to verify that this quantity is at
least $1-\delta/2$.  This completes the proof of the theorem.
\end{proof}

\subsection{Efficient Calculation of Marginal Gain}\label{sec:app:marginal}
A naive implementation of calculating the marginal gain $f_A(S \cup v) - f(S)$ takes $O(mk^2 + kmn_A)$ floating-point operations (FLOPs) where $|S| = O(k)$. The first term is from performing the Grahm-Schmidt orthonormalization of $S$, and the second term is from calculating the projection of each of $A$'s columns onto span$(S)$.

However, it is possible to significantly reduce the marginal gain calculations to $O(mn_A)$ FLOPs in $\greedy$ by introducing $k$ updates, each of which takes $O(mn_A + mn_B)$ FLOPs. This idea was originally proven correct by \cite{Farahat2}, but we discuss it here for completeness.

The simple yet critical observation is that $\greedy$ permanently keeps a column $v$ once it selects it. So when we select $v$, we immediately update all columns of $A$ and $B$ by removing their projections onto $v$. This allows us to calculate marginal gains in future iterations without having to consider $v$ again.

\subsection{Proof of Theorem~\ref{thm:core-set-2}} \label{app:core-set-2}
%\begin{proof}[Proof of Theorem~\ref{thm:core-set-2}]
%For any $\eps <1$, the expected value of the solution of the $r$-epoch $\distgreedy$ algorithm, for $r = O(\kappa(\opt)/\epsilon)$, is at least $(1-\eps)f(OPT)$.
%For any $r \geq 1$, the expected value of solution of multi-epoch greedy algorithm is at least $[1 - (1 - \Omega(\sigma_{min}(OPT)))^r] f(OPT)$. In particular, with $r = O(1/\sigma_{min}(OPT))$ rounds, the expected approximation factor becomes at least any constant (say $0.9$), and with $O(log(1/\epsilon)/\sigma_{min}(OPT))$ rounds the expected value of the solution is at least $(1 - \epsilon)f(OPT)$. 
%\end{proof}
\begin{proof}
  Let $C^t$ be the union of first $t$ solutions: $\cup_{j=1}^t S^{j}$. The main observation is that to compute $f^{t+1}(V)$, we can think of first subtracting off the components of the columns of $A$ along $C^t$ to obtain $A'$, and simply computing $f_{A'}(V)$. Now, a calculation identical to the one in Eq.~\eqref{eq:dotprod-lb2} followed by the ones in Eq.~\eqref{eq:start}-\eqref{eq:temp8} (to go from one vector to the matrix) implies that $f_{A'}(\opt) \ge \big( \sqrt{f_{A}(\opt)} - \sqrt{f_{A}(C^t)}\big)^2$.  Now we can complete the proof as (Theorem~\ref{thm:greedy-main}).
\iffalse
For any $t \geq 1$, we use $OPT$ as the benchmark, and note that $f^t(OPT) = f^t(OPT \cup A^{t-1}) \geq f(OPT)$. On the other hand, $f(\emptyset)$ is equal to $f(A^{t-1})$. So there is a gap of at least $f(OPT) - f(A^{t-1})$ to be exploited.  
Using Theorem~\ref{thm:core-set-1}, we know that in epoch $t$ we find a set $S^t$ with expected $\E[f^t(S^t)]$ at least $\Omega(\sigma_{min}(OPT) (f(OPT) - f(A^{t-1})))$. This can be rewritten as: 

$$
\E[f(A^t)] \geq (1-\sigma)\E[f(A^{t-1})] + \sigma f(OPT)
$$ 

where $\sigma$ is $\Omega(\sigma_{min}(OPT))$. By monotonicity of $f$ and induction, we can prove that $\E[f(A^t)]$ is at least $[1 - (1 - \Omega(\sigma_{min}(OPT)))^r] f(OPT)$ which completes the proof. 
\fi
\end{proof}

\subsection{Necessity of Random Partitioning} \label{app:rand-part}
We will now make the intuition formal. We consider the $\gcss(A, B, k)$ problem, in which we wish to cover the columns of $B$ using columns of $A$. Our lower bound holds not just for the greedy algorithm, but for any {\em local} algorithm we use -- i.e., any algorithm that is not aware of the entire matrix $B$ and only works with the set of columns it is given and outputs a poly$(k)$ sized subset of them.

\begin{thm} \label{thm:rand-part}
For any square integer $k \geq 1$ and constants $\beta, c >0$, there exist two matrices $A$ and $B$, and a partitioning of $(B_1, B_2, \dots, B_\ell)$ with the following property. Consider any local, possibly randomized, algorithm that takes input $B_i$ and outputs $O(k^{\beta})$ columns $S_i$.  Now pick $ck$ elements $S^*$ from $\cup_i S_i$ to maximize $f_A(S^*)$. Then, we have the expected value of $f_A (S^*)$ (over the randomization of the algorithm) is at most $O\left( \frac{c\beta \log k}{\sqrt{k}} \right) f_A(\opt_k)$.
\end{thm}

\begin{proof}
Let $k = a^2$, for some integer $a$.  We consider matrices with $a^2 + a^3$ rows.  Our target matrix $A$ will be a single vector containing all $1$'s.  The coordinates (rows) are divided into sets as follows: $X = \{1, 2, \cdots, a^2\}$, and for $1 \leq i \leq a^2$, $Y_i$ is the set $\{a^2 + (i-1) \cdot a + 1, i \cdot a^2 + (i-1) \cdot a + 2, \cdots, a^2 + i \cdot a\}$.  Thus we have $a^2 + 1$ blocks, $X$ of size $a^2$ and the rest of size $a$.

Now, let us describe the matrix $B_i$ that is sent to machine $i$.  It consists of all possible $a$-sized subsets of $X \cup Y_i$.\footnote{As described, it has an exponential in $a$ number of columns, but there are ways to deal with this.}  Thus we have $\ell = a^2$ machines, each of which gets $B_i$ as above.

Let us consider what a local algorithm would output given $B_i$.  Since the instance is extremely symmetric, it will simply pick $O(k^{\beta})$ sets, such that all the elements of $X \cup Y_i$ are {\em covered}, i.e., the vector $A$ restricted to these coordinates is spanned by the vectors picked. But the key is that the algorithm cannot distinguish $X$ from $Y_i$! Thus we have that any set in the cover has at most $O(\beta \log a)$ overlap with the elements of $Y_i$.\footnote{To formalize this, we need to use Yao's minmax lemma and consider the uniform distribution.}

Now, we have sets $S_i$ that all of which have $O(\beta \log a)$ overlap with the corresponding $Y_i$.  It is now easy to see that if we select at most $ck$ sets from $\cup_i S_i$, we can cover at most $c a^2 \log a$ of the coordinates of the $a^3$ coordinates in $\cup_i Y_i$. The optimal way to span $A$ is to pick precisely the indicator vectors for $Y_i$, which will cover a $(1-(1/a))$ fraction of the mass of $A$.  Noting that $k = a^2$, we have the desired result.
\iffalse
We construct the set of columns $A_i$ based on $X$ and $Y_i$ as follows. 
We put $a^2 + a^3$ rows in $A$ one for each number in $X \cup Y_1 \cup Y_2 \cdots Y_{a^2}$.
We set the number of machines $m$ to be equal to $k = a^2$.
For any size $a$ subset $S \subset X \cup Y_i$, we put an indicator column $1_S$ in $A_i$ with $a$ ones in entries that belong to $S$, and zeros elsewhere. 
Finally we set $B$ to be just one column with all its entries equal to one. The symmetry we observe in each machine makes any algorithm unable to distinguish between entries in $X$ and entries in $Y_i$. Therefore in expectation every selected column will have at most $log(a)$ non-zero entries in $Y_i$. Therefore we can say with high probability each column in $S_i$ has at most $log(a)$ non-zero entries in $Y_i$ since $S_i$ has size at most $k^{\beta}$ (polynomial in $k$). So in the pool of selected columns $\cup_{i=1}^m S_i$ with high probability each column has at most $log(a) \leq log(k)$ non-zero entries in $\cup_{i=1}^m Y_i$. So any subset of  $c \cdots k$ columns among the selected columns will not cover more than $c \cdot k log(k)$ entries of the $a^3 = k\sqrt{k}$ entries of $\cup_{i=1}^m Y_i$. The proof completes by observing that $1-\frac{1}{k}$ fraction of entries of $B$ are in $\cup_{i=1}^m Y_i$. 
\fi
\end{proof}

\end{document}